\documentclass[a4paper,11pt]{article}
\usepackage[margin=2cm,top=1in,bottom=2cm,a4paper]{geometry}
\usepackage[utf8]{inputenc}
\usepackage{amsmath, amsthm, enumerate, graphics, psfrag, amssymb}
\usepackage{graphicx}
\usepackage{color}
\usepackage{booktabs}
\usepackage{array}
\usepackage{tikz}
\usepackage{hyperref}
\usepackage{cleveref}
\usetikzlibrary{patterns}
\newcommand{\opt}{{\rm OPT}}

\newcommand{\calb}{{{{\cal B}}}}
\newcommand{\calm}{{{{\cal M}}}}
\newcommand{\calp}{{{{\cal P}}}}
\newcommand{\cals}{{{{\cal S}}}}
\newcommand{\area}{{\it a}}
\newcommand{\side}{\mbox{{\rm side}}}
\newcommand{\nfdh}{\mbox{{\rm NFDH}}}
\newcommand{\Oh}{\mathrm{O}}
\newcommand{\falpha}[2]{\alpha \langle {#1}, {#2} \rangle}
\newcommand{\fgamma}[2]{\gamma \langle {#1}, {#2} \rangle}

\newenvironment{pascal}%
         {\begin{list}%
                {}%
                {%
                 \setlength{\partopsep}{1mm}%
                 \setlength{\topsep}{0mm}%
                 \setlength{\parsep}{0mm}%
                 \setlength{\leftmargin}{10mm}%
                }%
          \item}%
         {\end{list}}

\newtheorem{theorem}{Theorem}[section]

\newtheorem{proposition}[theorem]{Proposition}
\newtheorem{lemma}[theorem]{Lemma}
\newtheorem{corollary}[theorem]{Corollary}
\usepackage{authblk}

\usepackage{comment} 

\begin{document}

\title{Improved Upper Bounds for Dynamic Bin Packing of General, Unit-Fraction, 
and Power-Fraction Squares\protect\footnote{\ This research was
  partially supported by
CNPq (Proc.~313146/2022-5, 
311892/2021-3, 
404315/2023-2, 404779/2025-5);  
and FAPESP (Proc.~%
2022/05803-3, 2023/13972-2, 2023/03167-5, 2025/27962-4).
}}
\author[1]{Miguel A. Mini\thanks{\texttt{m.alessandro.mini@gmail.com}}}
\author[1]{Fl\'avio K. Miyazawa\thanks{\texttt{fkm@ic.unicamp.br}}}
\author[1]{Gabriel M. Silva\thanks{Corresponding author: \texttt{g247113@dac.unicamp.br}}}
\author[2]{Yoshiko Wakabayashi\thanks{\texttt{yw@ime.usp.br}}}
\affil[1]{Institute of Computing, University of Campinas}
\affil[2]{Institute of Mathematics, Statistics and Computer Science,  University of S\~ao Paulo}
\date{}
\maketitle

\begin{abstract}
This paper presents significant upper-bound improvements for dynamic 2D square bin packing, where square items arrive and depart over time and the objective is to minimize the peak number of concurrent active unit bins. In our model, repacking is permitted only within a destination bin upon item arrival; migration between active bins is strictly forbidden. By introducing a streamlined two-list algorithm and proving a tight $5/16$ occupied-area bound for Next-Fit Decreasing Height, we reduce the upper bound on the asymptotic competitive ratio for arbitrary squares from $4.2154$ down to $3.918$, breaking a longstanding theoretical ceiling. For restricted variants, we establish asymptotic competitive ratios of at most $3.356$ for unit-fraction side lengths and $2.211$ for power-fraction side lengths.

 \end{abstract}
 
\section{Introduction}\label{sec:intro}

The dynamic square bin packing problem is an online multidimensional packing problem in which square items arrive and depart over time. The primary objective is to minimize the peak number of active bins required simultaneously throughout the entire packing process. We adopt the model introduced by Epstein and Levy~\cite{EpsteinL10}, where departure times are unknown to the online algorithm and item migration between bins is forbidden. However, upon the arrival of a new item, the algorithm may repack items within their destination bin to exploit space vacated by previous departures, enabling internal restructuring via algorithms such as Next-Fit Decreasing Height ($\nfdh$). Repacking is prohibited upon item departures.

In the one-dimensional dynamic setting, introduced by Coffman, Garey, and Johnson~\cite{CoffmanGJ83}, a modified First-Fit algorithm achieves an asymptotic competitive ratio of at most $2.788$, alongside an online lower bound of $2.388$. They also showed a lower bound of $2.5$ when the offline adversary is allowed to repack items (even moving them to other bins). Chan, Wong, and Yung~\cite{ChanWY09} strengthened this lower bound by showing that $2.5$ holds even against an offline adversary who is not allowed to repack, which Wong, Yung, and Burcea~\cite{wong20128} later improved to $8/3$.

Epstein and Levy~\cite{EpsteinL10} initiated the study of $d$-dimensional dynamic bin packing for~$d \ge 2$. For two-dimensional squares, they gave an algorithm with an asymptotic competitive ratio of $4.2154$ and established a lower bound of $2.2307$. They also derived upper and lower bounds for rectangles ($8.5754$ and $3.7$), cubes ($5.37037$ and $2.117$), and general $d$-dimensional items ($2 \times 3.5^d$ and $d+1$). Subsequently, Wong and Yung~\cite{WongY10} improved the upper bounds for general $d$-dimensional items to $7.788$ ($d=2$), $22.788$ ($d=3$) and $3^d$ ($d\geq 4$), but for general squares, $4.2154$ has remained the best-known upper bound.

A prominent restriction considers \emph{unit-fraction items}, whose side lengths take the form $1/k$ for $k \in \mathbb{Z}_{\ge 1}$. This structure enables tighter density guarantees and often serves as a stepping stone toward the general case. Bar-Noy, Ladner, and Tamir~\cite{BLT2007} analyzed standard online 1D unit-fraction packing, while Chan, Lam, and Wong~\cite{CHAN2008} proved that Any-Fit algorithms are $3$-competitive and First-Fit achieves a ratio of $2.4942$ against a $2.45$ lower bound. Han et al.~\cite{HAN2010} generalized the analytical framework for First-Fit under maximum item size restrictions.

Extending unit-fraction packing to higher dimensions, Burcea, Wong, and Yung~\cite{BWY2013} developed specialized classification schemes for two- and three-dimensional items, including \emph{power-fraction items} with side lengths $2^{-k}$ for $k \in \mathbb{Z}_{\ge 0}$. Their First-Fit strategies yielded competitive ratios of $6.7850$ (2D) and $21.6108$ (3D) for unit-fraction items, and $6.2455$ for 2D power-fraction items. When restricted to unit-fraction squares, their analysis implicitly yields a $3.9654$-competitive algorithm, improving on the $4.2154$ general bound.

\textbf{Contributions.} We present improved online algorithms and tightened competitive bounds for dynamic 2D square packing across three settings:

\begin{itemize}
   \item \textbf{General Squares:} We achieve an asymptotic competitive ratio of at most $3.918$, improving the $4.2154$ bound of Epstein and Levy~\cite{EpsteinL10}. While Epstein and Levy split the item sequence into three sublists, our improvement stems from a streamlined classification into two sublists, paired with a tighter lower bound on bin area occupancy through a refined analysis of $\nfdh$.
   \item \textbf{Unit-Fraction Squares:} We improve the upper bound from $3.9654$ to $3.356$.
   \item \textbf{Power-Fraction Squares:} We improve the competitive ratio from $2.4842$ (inherited from 1D unit-fraction dynamic packing~\cite{HAN2010}) to $2.211$.
\end{itemize}

For both restricted settings, the tighter bounds follow from exact computations of the minimum occupied area in levels that cannot accommodate a given item. As the $2.2307$ lower bound of Epstein and Levy~\cite{EpsteinL10} uses unit-fraction square inputs, it applies directly to both general and unit-fraction variants. Our results reduce the competitive gap from $1.89$ to $1.76$ for general squares, and from $1.78$ to $1.50$ for unit-fraction squares.

\textbf{Organization.} \Cref{sec:prelim} provides technical preliminaries. \Cref{sec:squares} presents the general dynamic square packing algorithm. \Cref{sec:uf,sec:pf} introduce the specialized analyses for unit-fraction and power-fraction items, respectively.

\section{Preliminaries}\label{sec:prelim}

In this section, we formalize the dynamic square bin packing model and establish key notation. Items (squares of side length at most~$1$) arrive and depart at arbitrary, unannounced times. Upon arrival, an item must be immediately packed into a unit bin; upon departure, it is removed from its bin. Crucially, our model permits \emph{intra-bin repacking}: when placing an item into a destination bin, an algorithm may rearrange existing items within that specific bin, but moving items between different bins is strictly prohibited.

To evaluate the performance of the algorithms, we use competitive analysis. Let $\sigma$ be an input sequence of arrivals and departures. A bin is \emph{active} at a given moment if it contains at least one item at that moment. We denote by $\opt(\sigma)$ the maximum number of active bins used at any moment by an optimal offline algorithm, which knows $\sigma$ in advance and may rearrange items arbitrarily across all bins at any time. For an online algorithm $\mathcal{A}$, let $\mathcal{A}(\sigma)$ be the maximum number of active bins used by $\mathcal{A}$ at any point in time. Algorithm $\mathcal{A}$ achieves an \emph{asymptotic competitive ratio} of $c$ if, for every input sequence $\sigma$,
\[
\mathcal{A}(\sigma) \le c \cdot \opt(\sigma) + b,
\]
where $b$ is an absolute constant independent of $\sigma$.

A primary building block in our work is the Next-Fit Decreasing Height ($\nfdh$) algorithm. Given a list $L$ of squares, $\nfdh$ first sorts them in non-increasing order of side length. It then packs the sorted squares side by side along the bottom of a unit bin (of width $1$ and height $1$) until a square no longer fits horizontally. At that point, the algorithm starts a new level along the horizontal line aligned with the top of the first—and tallest—square of the current level, continuing the packing process above this line. If a new level exceeds height $1$, we say that $\nfdh$ \emph{fails} to pack $L$ into a unit bin. 

However, in the analysis, it is convenient to speak of levels even when $\nfdh$ fails to pack a list~$L$.  
For that, consider that we run the procedure above in a vertical strip of width~$1$ and unbounded height, and produce a packing, say $\mathcal{P}(L)$.  
We say that such a packing consists of a sequence of levels called the \emph{level decomposition} of $\mathcal{P}(L)$. Because every square has side length at most $1$, this decomposition is well-defined for every list $L$. Consequently, $\nfdh$ packs $L$ into a single unit bin if and only if the heights of the levels in its level decomposition sum to at most~$1$.

Throughout this paper, $\area(x)$ denotes the area of an item $x$, and for any set $X$ of items, $\area(X) = \sum_{x \in X} \area(x)$. At any point in time, $\opt(\sigma)$ is clearly bounded below by the total area of all active items present. Henceforth, whenever we refer to a bin, it means a unit bin.

The following lemma will be used across the paper:


\begin{lemma}\label[lemma]{lemma:mono}
  Let $X = \{x_1,\dots,x_n\}$ and $Y=\{y_1,\dots, y_m\}$ be multisets of squares with side lengths $1 \ge \side(x_1) \ge \side(x_2) \ge \cdots \ge \side(x_n)$ and $1 \ge \side(y_1) \ge \side(y_2) \ge \cdots \ge \side(y_m)$. If $n \le m$, $\side(x_i) \le \side(y_i)$ for $i = 1, \ldots, n$, and $\nfdh$ can pack $Y$ into a unit bin, then $\nfdh$ can pack $X$ into a unit bin.
\end{lemma}

\begin{proof}
 Let $\mathcal{P}(X)$ (resp. $\mathcal{P}(Y)$) be the $\nfdh$ packing of $X$ (resp.\ $Y$) into a vertical strip of width~1 (resp. unit bin). For $i \ge 0$, let $r_i$ (resp.\ $s_i$) be the number of squares in the first $i$ levels of $\mathcal{P}(X)$ (resp. $\mathcal{P}(Y)$), with $r_0 = s_0 = 0$. Then the $(i+1)$-st level of $\mathcal{P}(X)$ has height $\side(x_{r_i+1})$.

 Let $k$ be the total number of levels of $\mathcal{P}(X)$. We may assume $k\ge 2$, otherwise, $X$ can be packed into a unit bin.  We claim that $s_i \le r_i$ for $i = 0,\dots, k - 1$.

  The claim holds for $i = 0$. Assume $s_i \le r_i$ for some $i \le k - 2$, and let $d := r_{i+1} - r_i$ be the number of squares in the $(i+1)$-st level of $\mathcal{P}(X)$. Because $i+1 \le k - 1$, the $(i+1)$-st level is not the last level of $\mathcal{P}(X)$, so $x_{r_{i+1}+1}$ exists and fails to fit in that level. Thus, $\sum_{j=1}^{d+1} \side(x_{r_i+j}) > 1$. Since side lengths are non-increasing and $s_i \le r_i$, we have $\side(y_{s_i+j}) \ge \side(y_{r_i+j}) \ge \side(x_{r_i+j})$ for every $j \ge 1$, which implies $\sum_{j=1}^{d+1} \side(y_{s_i+j}) > 1$. Hence, the $(i+1)$-st level of $\mathcal{P}(Y)$ contains at most $d$ squares, giving $s_{i+1}\le s_i + d \le r_i + d = r_{i+1}$, which completes the induction.

Since $s_{k-1} \le r_{k-1} < n \le m$, the first $k-1$ levels of $\mathcal{P}(Y)$ do not contain all items of $Y$, so $\mathcal{P}(Y)$ has at least $k$ levels. Using the hypotheses on the  side lengths, and the fact that $\nfdh$ packs $Y$ into a unit bin, we have 
\[
  \sum_{i=0}^{k-1} \side(x_{r_i+1})
  \;\le\; \sum_{i=0}^{k-1} \side(y_{r_i+1})
  \;\le\; \sum_{i=0}^{k-1} \side(y_{s_i+1})
  \;\le\; 1 .
\]
The first summation corresponds to the total height of the packing $\mathcal{P}(X)$, so $\nfdh$ packs $X$ into a unit bin.
\end{proof}


We obtain two immediate corollaries.

\begin{corollary}\label[corollary]{cor:enlarge}
Let $X$ be a multiset of squares that $\nfdh$ fails to pack into a single bin, let $x \in X$, and let $y$ be a square with $\side(y) \ge \side(x)$.  Then $\nfdh$ fails to pack $Y = (X \setminus \{x\}) \cup \{y\}$ into a single bin.
\end{corollary}

\begin{proof}
We prove the contrapositive. Suppose $\nfdh$ packs $Y$ into a single bin. Since $\side(x) \le \side(y)$, $X$ is obtained from $Y$ by replacing $y$ with a smaller or equal item $x$. By \Cref{lemma:mono}, $\nfdh$ can pack $X$ into a single bin, contradicting the hypothesis.
\end{proof}

\begin{corollary}\label[corollary]{cor:subset}
Let $X$ be a multiset of squares that $\nfdh$ packs into a single bin. Then $\nfdh$ packs every $X' \subseteq X$ into a single bin.
\end{corollary}

\begin{proof}
    This follows directly from \Cref{lemma:mono} with $X$ playing the role of $Y$ and $X'$ the role of $X$. Its hypotheses hold: $|X'| \le  |X|$, and since $X' \subseteq X$, the $i$-th largest side in $X'$ is at most the $i$-th largest side in $X$ for every $i$.
\end{proof}

Note that, by this last corollary, any algorithm that adds a square to a bin only when $\nfdh$ succeeds maintains the invariant that each of its bins can be packed by $\nfdh$: the invariant is preserved on an arrival by construction, and on a departure because the remaining squares form a subset of a multiset that $\nfdh$ packs, so \Cref{cor:subset} applies.

\section{Dynamic Square Bin Packing}\label{sec:squares}  

We now introduce our proposed algorithm, which we denote by SMB, a modification of the algorithm SQP, presented by Epstein and Levy~\cite{EpsteinL10}.  While the algorithm SQP subdivides the input items into three lists, the algorithm SMB subdivides the input items into two lists: a list $L_1$ consisting of {\it small} or {\it medium} squares, and a list $L_2$ consisting solely of {\it big} squares.  We say that a square is {\it small} if its side is at most $\frac{1}{3}$, it is {\it medium} if its side is larger than $\frac{1}{3}$ and at most $\frac{1}{2}$, and it is {\it big} if its side is larger than $\frac{1}{2}$.

There is a unique (optimal) way to pack the list $L_2$, since a bin holds at most one big square. As we will see, our algorithm opens a new bin only when all previously opened bins are occupied, so if it opens $N$ bins there is a moment at which $N$ big squares are present, and any packing uses at least $N$ bins at that moment. Hence, the algorithm for $L_2$ has an asymptotic competitive ratio of 1.

It then remains to give an algorithm for $L_1$. We will denote by SM the algorithm to pack $L_1$, to be described in what follows.

Roughly speaking, the algorithm SM manages two classes of bins: a class $\cals$ consisting of those bins which contain or contained (at some moment) at least one small square, and a class $\calm$ consisting of the remaining bins (that is, bins which do not contain and never contained a small square). The idea behind this algorithm is to pack a small (resp.\ medium) square ---with the $\nfdh$ algorithm--- giving priority to an existing bin in $\cals$ (resp.\ $\calm$).  When such a packing is not possible, then a new bin is opened.

As will become clear from the description below, we note that a bin that initially belongs to $\calm$ may receive a small square and, when this happens, it becomes a bin of $\cals$ and remains so (even if later all its small squares depart). Thus, ``belongs to $\cals$'' is a permanent status, but ``belongs to $\calm$'' may be a temporary status for a bin.

Throughout this section, ``packing a square $s$ into bin $B$ with $\nfdh$'' means applying $\nfdh$ to the multiset consisting of $s$ together with the squares currently in $B$. If $\nfdh$ can successfully pack the items into one bin, we say that it succeeds; otherwise, it fails. Note that this is permitted by the intra-bin repacking rule.

\medskip

\begin{pascal}
\hrule\vspace{1mm}
{\bf Algorithm} SM
\vspace{1mm}
\hrule\vspace{1mm}
{\bf Input:} A sequence $L$ of arriving or departing squares, each of side at most $\frac{1}{2}$.

{\bf Output:} For each arriving square, the bin in which it is allocated. 

\vspace*{2mm}

\item[1]\quad $\cals \leftarrow \emptyset$\,; $\calm \leftarrow \emptyset$. 

\item[2]\quad \textbf{while} the sequence $L$ is not finished
\item[3]
\quad\quad\quad let $s$ be the next square in the sequence $L$. 

\item[4]
\quad\quad\quad \textbf{if} $s$ is a departing square, \textbf{then} remove $s$ from the bin in which $s$ is placed.

\item[5]
\quad\quad\quad \textbf{if} $s$ is an arriving square, 

\item[6]
\quad\quad\quad\quad \textbf{if} $s$ is {\it small}, \textbf{then} ---if possible--- pack $s$ with the
$\nfdh$ algorithm into a bin $B$ \\
\hspace*{15.5mm} in $\cals \cup \calm$, giving priority to a bin in $\cals$. 
If $B$ is a bin in $\calm$, then remove $B$ from \\
\hspace*{15.5mm} $\calm$ and set $\cals\leftarrow \cals\cup\{B\}$. 
If $\nfdh$ fails to pack $s$ into a bin in $\cals\cup\calm$, 
then open \\
\hspace*{15.5mm} a new bin~$B$, pack $s$ into $B$ and set $\cals \leftarrow \cals \cup \{B\}$.
Output $B$. 

\item[7]
\quad\quad\quad\quad \textbf{if} $s$ is {\it medium}, \textbf{then} ---if possible--- 
 pack $s$ with the $\nfdh$ algorithm into a bin \\
 \hspace*{15.5mm} in $\calm \cup \cals$, giving priority to a bin in $\calm$.  
   If this is not possible,  open a new bin~$B$, \\
 \hspace*{15.5mm} pack $s$ into $B$ and set $\calm\leftarrow \calm \cup\{B\}$. Output $B$. 

\vspace{2mm}
\hrule\vspace{1mm}
\end{pascal}

\medskip

We use the following theorem and lemma:
\begin{theorem}[Meir and Moser~\cite{MeirM68}]\label{NewRef1}
  Every multiset $X$ of squares whose largest square has side $\ell$ can be packed by \nfdh\ in a rectangle with dimension $(b,h)$ if $\ell \leq \min\{b, h\}$ and $\area(X) \leq \ell^2 + (b-\ell)(h-\ell)$.
\end{theorem}

\begin{lemma}[Fernandes et al.~\cite{FernandesFMW19}]\label{lequatronono}
  Let $y$ be a small square and $X$ be a multiset of small or medium squares.  If $\nfdh$ cannot pack $X\cup\{y\}$ in one bin, then $\area(X)\geq \frac{4}{9}$.
\end{lemma}

The next lemma will be useful to prove that the asymptotic competitive ratio of algorithm SM is at most $2.918$.

\begin{lemma}\label{lecincodezesseis} 
  Let $X\cup \{y\}$ be a multiset of squares with side at most $\frac{1}{2}$. If $\nfdh$ can pack $X$ in one bin, but cannot pack $X\cup\{y\}$ in one bin, then $\area(X) > \frac{5}{16}$.
\end{lemma}

\begin{proof} 
  Let $X$ and $y$ be as stated in the lemma. Note that if $\nfdh$ cannot pack the multiset $X\cup\{y\}$ in one bin, then in particular it cannot pack the multiset $X$ together with a square of side $\frac{1}{2}$, by \Cref{cor:enlarge}. Therefore, it suffices to prove a lower bound for $\area(X)$ when $y$ has side $\frac{1}{2}$.

  Since any square in $X$ has side at most $\frac{1}{2}$, \nfdh\ forms at least two levels before failing. Consider the packing in the first level obtained with $\nfdh$, and suppose w.l.o.g. that $y$ is the first square packed in this level. Let $X^-$ be the multiset of squares in $X\cup\{y\}$ without the squares packed in the first level, and let $x\in X^-$ be the first square packed in the second level. As $\nfdh$ could not pack $X^-$ above the first level, we have from Theorem~\ref{NewRef1} that $\area(X^-) > x^2+(\frac{1}{2}-x)(1-x)$. The first level contains at least one square of $X$ whose side is at least that of $x$. Indeed, since $y$ has side $\frac{1}{2}$ and every square in $X$ has side at most $\frac{1}{2}$, the square $x$ would otherwise fit in the first level alongside $y$, contradicting the fact that $\nfdh$ opened a new level. Hence, by the ordering of $\nfdh$, we get that

\begin{align*}
 \area(X) &\geq \area(x) + \area(X^-) > x^2 + x^2 + \left(\frac{1}{2}-x\right)(1-x) \\
          &= 3x^2 - \frac{3x}{2} + \frac{1}{2} \geq \frac{5}{16}.
\end{align*}
The last inequality follows from the fact that the minimum of $f(x)=3x^2-\frac{3x}{2}+\frac{1}{2}$ is attained at $x=\frac{1}{4}$, and $f(\frac{1}{4})=\frac{5}{16}$.
\end{proof}

We observe that the fraction $\frac{5}{16}$ of Lemma~\ref{lecincodezesseis} is the best possible. Indeed, if we take $X$ as the multiset of five squares, two with side $\frac{1}{4}+\varepsilon$ and three with side $\frac{1}{4}$, and $y$ a square with side $\frac{1}{2}$, we have that $\area(X)$ can be made arbitrarily close to $\frac{5}{16}$. We thus obtain the following proposition:

\begin{proposition}
\label[proposition]{prop:tight516}
For every $\varepsilon > 0$ there is a multiset $X$ of squares of side at most $\frac{1}{2}$ such that $\nfdh$ packs $X$ in one bin but cannot pack $X \cup \{y\}$, where $y$ has side $\frac{1}{2}$, and $\area(X) < \frac{5}{16} + \varepsilon$.
\end{proposition}

\begin{proof}
  Let $X$ be a multiset of 5 squares, two of side $\frac{1}{4} + \varepsilon/2$ and three of side $\frac{1}{4}$. Note that $\nfdh$ cannot pack $X$ together with an item of side $\frac{1}{2}$, as there would be no space for the last item of side $\frac{1}{4}$. Since the bound $\frac{5}{16} + \varepsilon$ only weakens as $\varepsilon$ grows, we may assume $\varepsilon \leq \frac{1}{4}$. Note that $\nfdh$ can pack $X$ by putting the two items of size $\frac{1}{4} + \varepsilon/2$ and one item of side $\frac{1}{4}$ in the first level (of width $\frac{3}{4} + \varepsilon \leq 1$) and the remaining items in the second level. Furthermore, note that $\area(X) = \frac{5}{16} + \varepsilon / 2 + \varepsilon^2/2 \leq \frac{5}{16} + \varepsilon$, when $0<\varepsilon \leq \frac{1}{4}$.
\end{proof}

\begin{figure}[ht]
\centering
\includegraphics[width=0.5\textwidth]{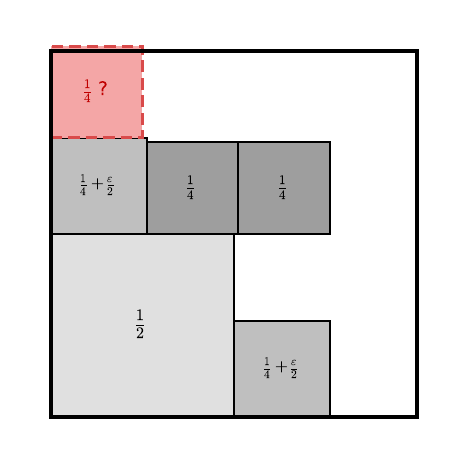}
\caption{Example for \Cref{prop:tight516}: A multiset $X$ of five squares (two of side $\frac{1}{4}+ \frac{\varepsilon}{2}$ and three of side $\frac{1}{4}$) that cannot be packed together with a square $y$ of side $\frac{1}{2}$. The area of $X$ can be made arbitrarily close to $\frac{5}{16}$.}
\label{fig:tight-lemma2}
\end{figure}

\begin{theorem}\label{teoSM}
The algorithm {\rm SM} has an asymptotic competitive ratio of at most $\frac{747}{256} \approx 2.918$.
\end{theorem}
\begin{proof}
  Let $\sigma$ be the input sequence for {\rm SM} and set $\calp := \cals \cup \calm$, the set of bins opened by the algorithm after processing $\sigma$. Here, $\opt$ denotes the maximum number of bins opened by an optimal offline algorithm for the input sequence $\sigma$. Recall that an offline algorithm knows its input completely before starting its processing. We would like to prove that $|\calp| \leq \frac{747}{256}\,\opt + \Oh(1)$.
  
  To show the desired inequality, we derive first two lower bounds for $\opt$. For that, let us analyze the last iteration $t$ at which the arrival of a square $s$ caused the union $\cals \cup \calm$ to grow by one bin, that is, it caused a (last) new bin $B$ to be opened. Let $k$ and $m$ be the number of bins in $\cals$ and $\calm$ respectively at the end of iteration $t$. Note that the number of bins used by {\rm SM} for $\sigma$ is at most $k+m$, as no more bins were opened after iteration $t$.

  There are two cases to analyze: at iteration $t$, the bin $B$ was included in $\cals$ (case 1) or $B$ was included in $\calm$ (case 2).

\medskip

\noindent {\bf Case 1:} $B$ was added to $\cals$ at iteration $t$.
\smallskip

In this case, clearly, the square $s$ that arrived at iteration $t$ is small.  Moreover, when $s$ arrived, there were $k-1$ bins in $\cals$ and $m$ bins in $\calm$. As $s$ is small and $s$ could not be packed with the $\nfdh$ algorithm into any of the $k-1$ bins in $\cals$, by Lemma~\ref{lequatronono}, each of these $k-1$ bins had an occupied area of at least~$\frac{4}{9}$.  Thus, the total area occupied by the squares in these bins at iteration $t$ was at least $\frac{4}{9}(k-1)$.

After trying to pack $s$ into each bin in $\cals\setminus\{B\}$, the algorithm~SM tried to pack $s$ into each bin in $\calm$. Thus, we know that $\nfdh$ failed to pack $s$ into such bins (as if $s$ could be packed into a bin in $\calm$, bin $B$ would not have been opened at iteration $t$). In this case, again, by Lemma~\ref{lequatronono}, the occupied area in each of these $m$ bins at iteration $t$ was at least~$\frac{4}{9}$. Thus, the total area occupied by the squares in these bins at iteration $t$ was at least $\frac{4}{9}m$.

It is immediate that the optimum is at least the total area of the squares before the arrival of square $s$. Thus,
\begin{equation}\label{eqdcase1}
 \opt\geq \frac{4}{9}(k-1) + \frac{4}{9}m.
\end{equation}


\noindent {\bf Case 2:} $B$ was added to $\calm$ at iteration $t$.
\smallskip

In this case, the square $s$ that arrived at iteration $t$ is medium and, when $s$ arrived, there were $k$ bins in $\cals$ and $m-1$ bins in $\calm$.
As $s$ could not be packed by $\nfdh$ into a bin in $\calm$, we can conclude that each of these $m-1$ bins had at least four medium squares (that is, squares of side greater than $\frac{1}{3}$). Indeed, if such a bin contained at most three medium squares, $\nfdh$ would place two of them in the first level and the remaining one together with $s$ in the second, as every medium square has side at most $\frac{1}{2}$. Thus, the squares in these bins at iteration $t$ had total area of at least~$\frac{4}{9}(m-1)$.
As $s$ could not be packed by $\nfdh$ into a bin in $\cals$, by Lemma~\ref{lecincodezesseis}, it follows that each of the $k$ bins in $\cals$ at iteration $t$ had occupied area of at least~$\frac{5}{16}$. As before, using an area argument, we have that
\begin{equation}\label{eqdcase2}
 \opt\geq \frac{5}{16}k + \frac{4}{9}(m-1).
\end{equation}

As one of the cases above must occur, we conclude that $\opt$ is at least the minimum of the right-hand side values of inequalities \eqref{eqdcase1} and \eqref{eqdcase2}, and therefore

\begin{equation}\label{eq1}
\opt \geq  \frac{5}{16}(k-1) +\frac{4}{9}(m-1).
\end{equation}

Now, let us derive another lower bound for the optimum.  Note that, as the algorithm runs, the cardinality of $\cals$ never decreases. Some bins of $\cals$ may become empty, but they remain in $\cals$.

As $|\cals|=k$ at the end of iteration~$t$, it means that at some iteration $r\leq t$ the number of bins in~$\cals$ was $k-1$ and grew to $k$. Clearly, this growth was caused by the arrival of a small square~$s$. As we have seen in the analysis of case~$1$, as $s$ could not be packed by the $\nfdh$ algorithm into any of the~$k-1$ bins that existed in $\cals$ when $s$ arrived, we can conclude that each of them had an occupied area of at least $\frac{4}{9}$ (by Lemma~\ref{lequatronono}).  Thus,
\begin{equation}\label{eq2}
\opt \geq  \frac{4}{9}(k-1).
\end{equation}

From (\ref{eq1}) and (\ref{eq2}), we have

\begin{equation}\label{eqd3}
\opt \geq \max\Bigl\{\frac{4}{9} (k-1)\,;\,\frac{5}{16}(k-1)+\frac{4}{9}(m-1)\Bigr\}.
\end{equation}

If $k = 0$, we get that $|\calp| \leq \frac{9}{4}\opt + 2$. Otherwise, set $k' := k-1$ and $m' := m-1$. Then, $|\calp|= k + m = k'+ m'+2$, and therefore, if $k'>0$ and $m'>0$, we have

$$\frac{|\calp|}{\opt} = \frac{ k'+ m'+ 2}{\opt} 
\leq  \frac{k'+ m'}{\max\{\frac{4}{9}k'\,;\,\frac{5}{16}k'+\frac{4}{9}m'\}}+\frac{2}{\opt} 
\leq 2.918+\frac{2}{\opt},$$  

\medskip

\noindent where the last inequality follows by writing $x = m'/k'$, so that the first term of the right-hand side becomes
$(1+x)/\max\{\frac49,\, \frac{5}{16}+\frac49 x\}$,
which increases on $[0,\frac{19}{64}]$ and decreases on $[\frac{19}{64},\infty)$. Its maximum $\frac{747}{256} < 2.918$ is attained at $m' = \frac{19}{64}k'$. Hence $|\calp| \leq \frac{747}{256}\opt + 2$.  If either $k'=0$ or $m'=0$, then $\opt\geq\frac{4}{9}\max\{k', m'\}$, and therefore $|\calp| \leq \frac{9}{4} \opt +2$.
\end{proof}

\medskip

Let us now describe the algorithm SMB, which dynamically packs squares of side at most~$1$.

\medskip

\begin{pascal}
\hrule\vspace{1mm}
{\bf Algorithm} SMB
\vspace{1mm}
\hrule\vspace{1mm}
{\bf Input:} A sequence $L$ of arriving or departing squares, each of side at most $1$.

{\bf Output:} For each arriving square, the bin in which it is allocated. 

\vspace*{2mm}

\item[1]\quad $\calb \leftarrow \emptyset$. 

\item[2]\quad Perform an online partition of $L$ into two subsequences: \\
\hspace*{8.5mm} let $L_1$ be the subsequence of $L$ with the {\it small}  and {\it medium} squares, and \\
\hspace*{8.5mm} let $L_2$ be the subsequence of $L$ with the {\it big} squares. 

\item[3]\quad Apply algorithm SM to list $L_1$. 

\item[4]\quad Pack each square $s$ of $L_2$ into an empty bin $B$ in $\calb$ if one exists; \\
\hspace*{8.5mm} if not, pack $s$ into a new empty bin $B$ and set $\calb \leftarrow \calb \cup\{B\}$. Output $B$. 
\vspace{2mm}
\hrule\vspace{1mm}
\end{pascal}

\medskip

The algorithms for the two lists $L_1$ and $L_2$ have asymptotic competitive ratios of $\frac{747}{256} \approx 2.918$ and $1$, respectively. We use the following observation of Epstein and Levy~\cite{EpsteinL10}: if the input sequence $\sigma$ is partitioned online into subsequences $\sigma_1, \dots, \sigma_k$ and each $\sigma_i$ is packed into its own set of bins by an algorithm with asymptotic competitive ratio $\alpha_i$, then the combined algorithm has asymptotic competitive ratio at most $\sum_{i=1}^{k} \alpha_i$, since an optimal offline packing of $\sigma_i$ never uses more bins than an optimal offline packing of $\sigma$.

\begin{theorem} The algorithm {\rm SMB} has an asymptotic competitive ratio of at most $\frac{1003}{256} \approx 3.918$.
\end{theorem}
\begin{proof}
  By the observation above, the asymptotic competitive ratio of {\rm SMB} is at most the sum of those of the algorithms for $L_1$ and $L_2$, namely $\frac{747}{256} + 1 = \frac{1003}{256} \approx 3.918$.
\end{proof}

\section{Unit-Fraction Items}\label{sec:uf}

In this section, we improve the results for the specific case when the squares have unit-fraction lengths, that is, lengths of the type $\frac{1}{k}$ for $k \in \mathbb{Z}^+$. To start, we define

\begin{equation*}
\falpha{x}{y} = \min \left\{\sum_{i=y}^x\frac{1}{i^2} n_i \;\bigg| \; 1- \frac{1}{y} - \frac{1}{x} < \sum_{i=y}^x\frac{1}{i} n_i \leq 1 - \frac{1}{y}\quad\mbox{and}\quad n_i\in \mathbb{Z}_{\geq 0}\right\}.
\end{equation*}

Intuitively, $\falpha{x}{y}$ represents the minimal total area that can be achieved by arranging squares of side length at least $\frac{1}{x}$ and at most $\frac{1}{y}$ side by side, such that the combined length of their arrangement fits within a segment of length $1 - \frac{1}{y}$ but cannot accommodate an additional square of side $\frac{1}{x}$.

This means the following. Consider a level of an $\nfdh$ packing whose first square has side $\frac{1}{y}$, and suppose that the next square considered by $\nfdh$ has side $\frac{1}{x}$ and does not fit in this level.  Every other square of the level has side at most $\frac{1}{y}$, since the first square of a level is its largest, and at least $\frac{1}{x}$, since $\nfdh$ considers squares in non-increasing order of side; their sides sum to at most $1 - \frac{1}{y}$, as the level has width at most $1$, and to more than $1 - \frac{1}{y} - \frac{1}{x}$, as otherwise the square of side $\frac{1}{x}$ would have fitted.  Hence their multiplicities form a feasible solution of the program above, and $\falpha{x}{y}$ is a lower bound on the total area of the squares of the level other than its first one.

Sometimes, it suffices to use a lower bound on this function; hence, we present the following lemma.

\begin{lemma}
\label[lemma]{lemma:inequality}
    $\falpha{x}{y} \geq (1 - \frac{1}{y} - \frac{1}{x})\frac{1}{x}$.
\end{lemma}

\begin{proof}

  Since a square with side $\frac{1}{x}$ does not fit, the occupied length is at least $ 1 - \frac{1}{y} - \frac{1}{x}$. Also, all the squares have height at least $\frac{1}{x}$ so the area is at least $(1 - \frac{1}{y} - \frac{1}{x})\frac{1}{x}$.
    
\end{proof}

However, this lower bound is sometimes insufficient, so we establish an important fact about $\falpha{x}{y}$ that lets us compute it faster.

\begin{lemma}
\label[lemma]{lemma:compute}
    Let $n_y, \dots, n_x$ be an optimal solution, that is, such that $\falpha{x}{y} = \sum_{i=y}^x\frac{1}{i^2} n_i$. For each $i \in \{y, \dots, x - 1 \}$, let $k_i$ be the smallest integer such that there exists an integer $j$ satisfying $\frac{j}{x} \geq \frac{k_i}{i}$ and $\frac{j}{x^2} < \frac{k_i}{i^2}$. Then $n_i < k_i$. 
\end{lemma}

\begin{proof}
    Suppose that there exists an $i$ such that $n_i \geq k_i$. Let $j^*$ be the smallest integer such that $\frac{j^*}{x} \geq \frac{k_i}{i}$ and $\frac{j^*}{x^2} < \frac{k_i}{i^2}$. Now, consider another solution where $n_x' = n_x + j^*$ and $n_i'= n_i - k_i$. By the definition of $j^*$, the area of this new solution is smaller than $\falpha{x}{y}$, so if we show that $1 - \frac{1}{y} - \frac{1}{x} < \sum_{i=y}^x\frac{1}{i} n_i'  \leq 1 - \frac{1}{y}$ we get a contradiction. In fact, since $j^*$ is the smallest integer such that  $\frac{j^*}{x} \geq \frac{k_i}{i}$ and $\frac{j^*}{x^2} < \frac{k_i}{i^2}$, we increase $\sum_{i=y}^x \frac{n_i}{i}$ by at most $\frac{1}{x}$. If $\sum_{i=y}^x \frac{n_i'}{i} \leq 1 - \frac{1}{y}$, we get the desired contradiction. Otherwise, if $\sum_{i=y}^x \frac{n_i'}{i} > 1 - \frac{1}{y}$, we can decrease $n'_x$ by one and get $ 1 -\frac{1}{y} - \frac{1}{x}< \sum_{i=y}^x \frac{n_i'}{i} \leq 1 - \frac{1}{y}$ and also obtain a contradiction.
\end{proof}

Note that this does not bound $n_x$, but we can easily note that $n_x \leq (1 - \frac{1}{y})x$, so we can enumerate each $n_i$ to find $\falpha{x}{y}$.
We will also use the following fact.

\begin{lemma}
\label[lemma]{lemma:increasing}
    Let $c, y > 0$ be fixed and let $h(x) = $ 
    $(1 - \frac{1}{y} - \frac{1}{x})\frac{1}{x} + \frac{1}{x^2} + (1- \frac{1}{x})(\frac{1}{c} - \frac{1}{x})$. Then $h(x)$ is increasing for $x > \max(c,y)$.
\end{lemma}

\begin{proof}
    Taking the derivative with respect to $x$ yields the following expression $\frac{cx - 2cy +xy}{cx^3y}$. Since we want it to be positive, it suffices that $cx -2cy + xy > 0$, so $x > \frac{2cy}{c+y}$. Moreover, $\frac{2cy}{c+y} \leq \frac{2cy}{2\min(c,y)} = \max(c,y)$, so $h'(x) > 0$ whenever $x > \max(c,y)$.
\end{proof}

Throughout this paper, we will need to compute some exact values of $\falpha{x}{y}$. We do this through exhaustive enumeration combined with \Cref{lemma:compute}. More specifically, for each $n_i$, with $i \in \{y, \dots, x - 1\}$, we bound $n_i < k_i$ and $n_x \leq \lfloor (1 - 1 / y)x \rfloor$. (Every $k_i$ is well defined, as any $k \geq i^2/(x(x - i))$ admits an integer $j$ with $j / x \geq k/i$ and $j/x^2 < k/i^2$.) The implementation additionally prunes by branch and bound, using that completing a partial solution of width $w$ costs area at least $(1 - 1/y - 1/x - w)/x$; this does not affect the value returned. The largest enumeration, which is used to compute $\falpha{42}{7}$, has a search space of 313\,286\,400 candidates a priori, which is reduced by the branch-and-bound to only 159 nodes, and runs in less than one second. We use exact fraction arithmetic; more specifically, we use the Python module \textit{fractions}. We note that the witness column of each table exhibits a feasible multiset attaining the stated value, so the upper bounds on $\falpha{x}{y}$ are verifiable; the enumeration serves as a certification that no feasible multiset has a smaller area. The implementation is available at \url{https://doi.org/10.5281/zenodo.22286701}. 

We are now ready to improve on Lemma \ref{lecincodezesseis}.

\begin{lemma}
\label[lemma]{lemma:ufhalf}
Let $X$ be a multiset of unit-fraction squares of side at most $\frac{1}{2}$, and let $y$ be a unit-fraction square of side at most $\frac{1}{2}$. If $\nfdh$ cannot pack $X \cup \{y\}$ in one bin, then $\area(X) \geq \frac{1022131}{2965284}$.
\end{lemma}

\begin{proof}

  Without loss of generality, let $y$ be a square with side $\frac{1}{2}$, by \Cref{cor:enlarge}. Also without loss of generality, we assume that the first item packed is $y$, which is possible since there is no item larger than $y$. Since the largest square also has side length at most $\frac{1}{2}$, $\nfdh$ creates at least two levels before failing. Let $\frac{1}{x}$ be the side of the first item packed in the second level. Furthermore, by \Cref{NewRef1}, $\nfdh$ can always pack squares of side at most $\frac{1}{x}$ into a rectangle of width $1$ and height $\frac{1}{2}$ if their area is at most $f(x) = \frac{1}{x^2} + (\frac{1}{2} - \frac{1}{x})(1 - \frac{1}{x})$; hence, a multiset that $\nfdh$ fails to pack has area greater than $f(x)$.

We can now partition $X$ into two multisets $X^+$ and $X^-$, where $X^+$ are the items packed into the first level and $X^- = X \setminus X^+$. As argued before, we get that $\area(X^+) \geq \falpha{x}{2}$ since $\nfdh$ sorts the items so every item in the first level has side at least $\frac{1}{x}$. Furthermore, $\area(X^-) > f(x)$ since $X^-$ cannot be packed into the remaining space. So we get the following inequality $\area(X) = \area(X^+) + \area(X^-) > \falpha{x}{2} + f(x)$. We can thus begin a case analysis using different values of $x$, as shown in \Cref{tab:falpha_2}.

\begin{table}[htbp]
\centering
\renewcommand{\arraystretch}{1.3} 
\begin{tabular}{@{\hspace{0.3cm}} c c @{\hspace{1.5cm}}c @{\hspace{1.5cm}} c @{\hspace{0.3cm}}}
\toprule
\textbf{$x$} & \textbf{$\falpha{x}{2}$} & \textbf{Witness} & \textbf{$\falpha{x}{2} + f(x)$} \\
\midrule
2 & $\frac{1}{4}$ & $1: \frac{1}{2}$ & $\frac{1}{2} = 0.5$\\
3 & $\frac{1}{9}$ & $1: \frac{1}{3}$ & $\frac{1}{3} \approx 0.333$\\
4 & $\frac{1}{9}$ & $1: \frac{1}{3}$ & $\frac{13}{36} \approx 0.361$\\
5 & $\frac{2}{25}$ & $2: \frac{1}{5}$ & $\frac{9}{25} = 0.36$\\
6 & $\frac{61}{900}$ & $1:\frac{1}{5}, 1:\frac{1}{6}$ & $\frac{28}{75} \approx 0.373$\\
7 & $\frac{3}{49}$ & $3: \frac{1}{7}$ & $\frac{19}{49} \approx 0.387$\\
8 & $\frac{81}{1568}$ & $1: \frac{1}{7}, 2: \frac{1}{8}$ & $\frac{155}{392} \approx 0.395$\\

\bottomrule
\end{tabular}
\caption{Values for $\falpha{x}{2}.$ The worst case happens when $x = 3$. The entry $c : \frac{1}{i}$ denotes $c$ squares of side $\frac{1}{i}$.}
\label{tab:falpha_2}
\end{table}

When $x \geq 9$, we can drop the $\falpha{x}{2}$ term and use only $f(x)$. Since $f(x)$ is an increasing function, for this last case we obtain that $\area(X) \geq f(9) = \frac{29}{81}\approx 0.358$.

The witness column of each table exhibits a feasible multiset attaining the stated value, which certifies the corresponding upper bound on $\falpha{x}{y}$. Now, since the worst case happens when $x = 3$, we see that $\area(X) \geq \frac{1}{3}.$

We now refine the case where $x = 3$. Suppose that there are only 2 levels. Since we cannot pack $X \cup \{y\}$, let $\frac{1}{x^+}$ be the size of the first item that $\nfdh$ fails to pack. Note that $\frac{1}{x^+}$ must be between $\frac{1}{5}$ and $\frac{1}{3}$, as otherwise it would have fit in a third level. We first note that $\falpha{3}{2}  + \frac{1}{9} = \frac{2}{9}$. If $x^+ = 3$, then $\area(X) \geq \frac{2}{9} + \falpha{3}{3} + \frac{1}{9} = \frac{5}{9} \approx 0.555$. If $x^+ =4$, then $\area(X) \geq \frac{2}{9} + \falpha{4}{3} + \frac{1}{16} = \frac{59}{144} \approx 0.409$. Finally, if $x^+ = 5$, then $\area(X) \geq \frac{2}{9} + \falpha{5}{3} + \frac{1}{25} = \frac{86}{225} \approx 0.382$. We can now evaluate what happens if there is another level, where $\frac{1}x'$ is the side of the first item packed in that level and, using a similar idea to the previous one, we have that $\area(X) \geq \frac{2}{9} + \falpha{x'}{3} + \frac{1}{x'^2} + (1 - \frac{1}{x'})(\frac{1}{6} - \frac{1}{x'})$. Results are shown in \Cref{tab:falpha_3}.

\begin{table}[htbp]
\centering
\renewcommand{\arraystretch}{1.3} 
\begin{tabular}{@{\hspace{0.3cm}} c c @{\hspace{1.5cm}}c @{\hspace{1.5cm}} c @{\hspace{0.3cm}}}
\toprule
\textbf{$x'$} & \textbf{$\falpha{x'}{3}$} & \textbf{Witness} & \textbf{$\frac{2}{9} + \falpha{x'}{3} + \frac{1}{x'^2} + (1 - \frac{1}{x'})(\frac{1}{6} - \frac{1}{x'})$} \\
\midrule
6 & $\frac{43}{450}$ & $1: \frac{1}{5}, 2: \frac{1}{6}$ & $\frac{311}{900} \approx 0.3456$\\
7 & $\frac{4}{49}$ & $4: \frac{1}{7}$ & $\frac{152}{441} \approx 0.3447$\\
8 & $\frac{241}{3136}$ & $3: \frac{1}{7}, 1:\frac{1}{8}$ & $\frac{9911}{28224} \approx 0.351$\\
9 & $\frac{337}{5184}$ & $1: \frac{1}{8}, 4: \frac{1}{9}$ & $\frac{67}{192} \approx 0.3490$\\
10 & $\frac{3}{50}$ & $6:\frac{1}{10}$ & $\frac{317}{900} \approx 0.3522$\\
11 & $\frac{171}{3025}$ & $4: \frac{1}{10}, 2: \frac{1}{11}$ & $\frac{9689}{27225} \approx 0.3559$\\
12 & $\frac{145}{2904}$ & $1: \frac{1}{11}, 6: \frac{1}{12}$ & $\frac{3097}{8712} \approx 0.3555$\\
13 & $\frac{8}{169}$ & $8: \frac{1}{13}$ & $\frac{545}{1521} \approx 0.3583$\\
14 & $\frac{89767}{2004002}$ & $1: \frac{1}{11}, 1: \frac{1}{13}, 6:\frac{1}{14}$  & $\frac{13005899}{36072036} \approx 0.3606$\\
15 & $\frac{1793}{44100}$ & $1: \frac{1}{14}, 8 : \frac{1}{15}$ & $\frac{3181}{8820} \approx 0.3607$\\
16 & $\frac{5}{128}$ & $10: \frac{1}{16}$ & $\frac{209}{576} \approx 0.3628$\\
17 & $\frac{67313}{1812608}$ & $1:\frac{1}{14}, 2:\frac{1}{16}, 7:\frac{1}{17}$ & $\frac{5943289}{16313472} \approx 0.3643$\\
18 & $\frac{1607}{46818}$ & $1: \frac{1}{17}, 10:\frac{1}{18}$ & $\frac{3793}{10404} \approx 0.3646$\\
19 & $\frac{12}{361}$ & $12:\frac{1}{19}$ & $\frac{1190}{3249} \approx 0.3663$\\
20 & $\frac{165729}{5216450}$ & $1: \frac{1}{17}, 3:\frac{1}{19}, 8:\frac{1}{20}$ & $\frac{1724524}{4694805} \approx 0.3673$\\
21 & $\frac{1747}{58800}$ & $1: \frac{1}{20}, 12:\frac{1}{21}$ & $\frac{9263}{25200} \approx 0.3676$\\
22 & $\frac{7}{242}$ & $14: \frac{1}{22}$ & $\frac{1607}{4356} \approx 0.3689$\\
23 & $\frac{313933369}{11291187600}$ & $1: \frac{1}{20}, 1:\frac{1}{21}, 2:\frac{1}{22}, 10:\frac{1}{23}$ & $\frac{1391632723}{3763729200} \approx 0.3697$\\
24 & $\frac{3991}{152352}$ & $1:\frac{1}{23}, 14:\frac{1}{24}$ & $\frac{28181}{76176} \approx 0.3699$\\

\bottomrule
\end{tabular}
\caption{Values for $\falpha{x'}{3}$ used in \Cref{lemma:ufhalf}. The worst case happens when $x' = 7$.}
\label{tab:falpha_3}
\end{table}

When $x' \geq 25$, we can again drop the $\falpha{x'}{3}$ term and use only $\frac{2}{9}  + \frac{1}{x'^2} + (1 - \frac{1}{x'})(\frac{1}{6} - \frac{1}{x'})$, to obtain that, in this case, $\area(X) \geq \frac{1943}{5625} \approx 0.3454$.

The worst case happens when $x' = 7$, so we get that $\area(X) \geq \frac{152}{441}$.

To finish the proof, let us refine the case where $x' = 7$. Note that $\falpha{7}{3} + \frac{1}{49} = \frac{5}{49}$. If there are only 3 levels, denote by $\frac{1}{x''}$ the side of the smallest item being packed. $x''$ can take any value in the set $\{7, 8,\dots, 41\}$, so the area is at least $\falpha{x''}{7} + \frac{2}{9} + \frac{5}{49} + \frac{1}{x''^2} \geq (1 - \frac{1}{7} - \frac{1}{x''})\frac{1}{x''} +  \frac{2}{9} + \frac{5}{49} + \frac{1}{x''^2}$, which is a decreasing function in $x''$, so the worst case happens when $x'' = 41$, that yields $\area(X) \geq \frac{6241}{18081}$. Now we have to consider the case when $x'' \geq 42$. We have that $\area(X) \geq \falpha{x''}{7} + \frac{2}{9} + \frac{5}{49} + (\frac{1}{42} - \frac{1}{x''})(1 - \frac{1}{x''}) + \frac{1}{x''^2}$. When $x'' = 42$, since $\falpha{42}{7} = \frac{29459}{1482642}$ we get that $\area(X) \geq \frac{1022131}{2965284}$. 

Otherwise, we can use Lemma \ref{lemma:inequality} to obtain that $\area(X) \geq (1 - \frac{1}{7} - \frac{1}{x''})\frac{1}{x''} +  \frac{2}{9} + \frac{5}{49} + (\frac{1}{42} - \frac{1}{x''})(1 - \frac{1}{x''}) + \frac{1}{x''^2}$. As noted in Lemma \ref{lemma:increasing}, this is an increasing function, so the worst case happens when $x'' = 43$, so we obtain that $\area(X) \geq \frac{281102}{815409}$. 
Comparing all the subcases, the worst case happens when $x'' = 42$, so we obtain that overall $\area(X) \geq \frac{1022131}{2965284}$.

\end{proof}

We can also adapt the other lemma to obtain:

\begin{lemma}
\label[lemma]{lemma:ufsmall}
Let $X$ be a multiset of unit-fraction small or medium squares, and let $y$ be a
unit-fraction small square. If $\nfdh$ cannot pack $X \cup \{y\}$ in one bin,
then $\area(X) \geq \frac{119498}{247107}$.
\end{lemma}

\begin{proof}

    We consider cases based on the number of squares of side length $\frac{1}{2}$ in $X$. If there are at least two of them, then $\area(X) \geq \frac{1}{2}$, and the result is valid. 
    Furthermore, if there is only one of them, we can assume without loss of generality that $y = \frac{1}{3}$, as per \Cref{cor:enlarge}. This reduces to the case in \Cref{lemma:ufhalf}, except with $X$ with one more item of side $\frac{1}{2}$ and one fewer of side $\frac{1}{3}$, so $\area(X) \geq \frac{1022131}{2965284} + \frac{1}{4} - \frac{1}{9} = \frac{119498}{247107}$. We can now analyze the case where there is no item with length $\frac{1}{2}$.

    We can follow a very similar approach to the previous lemma's proof, by defining $g(x) = \frac{1}{x^2} + (\frac{2}{3} - \frac{1}{x})(1 - \frac{1}{x})$, where we get the inequality of the form $\area(X) \geq \falpha{x}{3} + g(x)$, where $\frac{1}{x}$ is again the side of the first item packed into the second level. Results are shown in \Cref{tab:falpha_3_2}.

\begin{table}[htbp]
\centering
\renewcommand{\arraystretch}{1.3} 
\begin{tabular}{@{\hspace{0.3cm}} c c @{\hspace{1.5cm}}c @{\hspace{1.5cm}} c @{\hspace{0.3cm}}}
\toprule
\textbf{$x$} & \textbf{$\falpha{x}{3}$} & \textbf{Witness} & \textbf{$\falpha{x}{3} + g(x)$} \\
\midrule
3 & $\frac{2}{9}$ & $2: \frac{1}{3}$ & $\frac{5}{9} \approx 0.556$\\
4 & $\frac{1}{8}$ & $2: \frac{1}{4}$ & $\frac{1}{2} = 0.5$\\
5 & $\frac{3}{25}$ & $3: \frac{1}{5}$ & $\frac{8}{15} \approx 0.533$\\
6 & $\frac{43}{450}$ & $1: \frac{1}{5}, 2: \frac{1}{6}$ & $\frac{27}{50} = 0.54$\\
7 & $\frac{4}{49}$ & $4:\frac{1}{7}$ & $\frac{27}{49} \approx 0.551$\\
8 & $\frac{241}{3136}$ & $3: \frac{1}{7}, 1: \frac{1}{8}$ & $\frac{5329}{9408} \approx 0.566$\\

\bottomrule
\end{tabular}
\caption{Values for $\falpha{x}{3}$ used in \Cref{lemma:ufsmall}. The worst case happens when $x = 4$.}
\label{tab:falpha_3_2}
\end{table}

Again, when $x \geq 9$, we can ignore the term $\falpha{x}{3}$ and obtain that $\area(X)\geq \frac{41}{81} \approx 0.506$. Comparing the cases, the worst case happens when there is one item of side $\frac{1}{2}$, so $\area(X) \geq \frac{119498}{247107}$.

\end{proof}

Combining both lemmas, we obtain the following bound:

\begin{theorem}
\label{thm:ufsm}
The algorithm {\rm SM} has an asymptotic competitive ratio of at most $\frac{3377129}{1433976} \approx 2.355$ when all items are unit-fraction squares with side at most $\frac{1}{2}$.
\end{theorem}

\begin{proof}
    The improved bounds from Lemmas \ref{lemma:ufhalf} and \ref{lemma:ufsmall} are:
    \begin{itemize}
        \item Lemma \ref{lemma:ufhalf}: $\area(X) \geq \frac{1022131}{2965284} \approx 0.34470$.
        
        \item Lemma \ref{lemma:ufsmall}: $\area(X) \geq \frac{119498}{247107} \approx 0.48359$.
    \end{itemize}

    Since the items are unit fractions, we note that the medium class of items consists only of squares of side $\frac{1}{2}$. Therefore, when a square does not fit into a bin in $\calm$, the occupied area in such a bin is $1$. This comes from the fact that, if there were at most 3 squares of side $\frac{1}{2}$ in such a bin, any item $s$ with side length $\leq \frac{1}{2}$ would be packed using $\nfdh$, which places two squares of side $\frac{1}{2}$ at the first level and the other $\frac{1}{2}$-square and $s$ on the second one. We can now proceed as in the proof of \Cref{teoSM} to obtain

    \begin{equation*}
        \opt \geq \max \left\{\frac{119498}{247107}(k -1),\, \min \left\{ \frac{119498}{247107}(k -1) + m, \, \frac{1022131}{2965284}k + m-1\right\} \right\}.
        \end{equation*}

\noindent Thus, 
    \begin{equation*}
        \opt \geq \max \left\{\frac{119498}{247107}(k - 1),  \,\frac{1022131}{2965284}(k - 1) + m-1 \right\}.
    \end{equation*}

\noindent  If $k = 0$, we obtain that $|\calp| \leq \opt + 2$. By setting $k' = k - 1, m' = m - 1$, we have that 

    $$\frac{|\calp|}{\opt} = \frac{ k'+ m'+ 2}{\opt} 
    \leq  \frac{k'+ m'}{\max \left\{\frac{119498}{247107}k',  \frac{1022131}{2965284}k' + m' \right\} }+\frac{2}{\opt} 
    \leq \frac{3377129}{1433976}+\frac{2}{\opt}.$$  

    The last inequality follows by observing that, by setting $x = \frac{m'}{k'}$
    in the inequality in the middle,  the ratio of the first part of this inequality  becomes $\frac{1 + x}{\max\{A, B + x\}}$, for $ A = \frac{119498}{247107}$ and $B = \frac{1022131}{2965284}$. This ratio increases on $[0, A-B]$ and decreases on $[A - B, \infty)$ as $B < 1$, and its  maximum value  $\frac{3377129}{1433976}$ is attained at $ x= A - B$, that is, $m' = \frac{5}{36}k'$. When $k' = 0$, we have that the ratio goes to 1. Finally, if $m' \leq 0$, then $\opt \geq \frac{119498}{247107}k'$ gives $|\calp| = k' + m' + 2 \leq \frac{247107}{119498}\opt + 2 \leq \frac{3377129}{1433976}\opt + 2$.
    
\end{proof}

The bound of Lemma~\ref{lemma:ufhalf} is best possible.

\begin{proposition}\label[proposition]{prop:tight}
For every $\varepsilon > 0$ there is a multiset $X$ of unit-fraction squares of
side at most $\frac{1}{2}$ such that $\nfdh$ packs $X$ into a single bin but
cannot pack $X \cup \{s\}$, where $s$ is a square of side $\frac{1}{2}$, and
$\area(X) \leq \frac{1022131}{2965284} + \varepsilon$.
\end{proposition}

\begin{proof}
Let $m \geq \lceil 1/\varepsilon \rceil$ with $m > 42$ and let $X$ consist of two squares of side ${1}/{3}$, five of side ${1}/{7}$, one of side ${1}/{41}$,
thirty-five of side ${1}/{42}$, and $m$ squares of side ${1}/{m}$. Then
$$\area(X) = \frac{1022131}{2965284} + \frac{1}{m} \leq
\frac{1022131}{2965284} + \varepsilon.$$

Running $\nfdh$ on $X \cup \{s\}$ produces the four levels shown in 
Figure~\ref{fig:tight-lemma34}, of heights
${1}/{2}$, ${1}/{3}$,
${1}/{7}$ and ${1}/{42}$,
which sum to exactly $1$. The fourth level
holds one square of side $\frac{1}{42}$ and then $\lfloor \frac{41m}{42}
\rfloor$ squares of side $\frac{1}{m}$; the remaining squares of side
$\frac{1}{m}$ cannot be placed, since the level has no width left and no
further level fits. Hence, $\nfdh$ fails on $X \cup \{s\}$. On the other hand, 
$\nfdh$ packs $X$ alone in height $\frac{1}{2} + \frac{1}{m} < 1$.
\end{proof}

\paragraph{Remark.} The bound in \Cref{lemma:ufsmall} is also the best possible. By considering the same $X$ as in the previous proposition but swapping one of the squares with side $\frac{1}{3}$ for a square of side $\frac{1}{2}$,  we get that $\area(X) \leq \frac{119498}{247107} + \varepsilon$.

\begin{figure}[ht]
\centering
\includegraphics[width=0.5\textwidth]{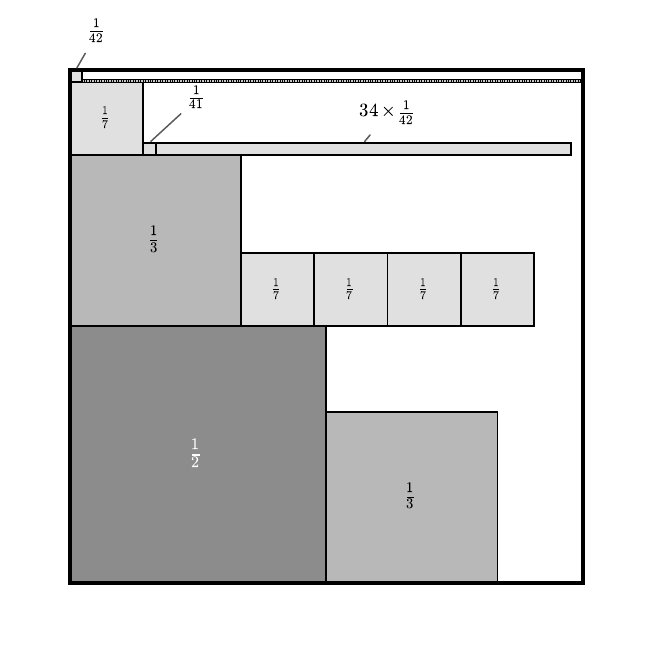}
\caption{Example for \Cref{prop:tight}: A multiset $X$ with 2 items of side $\frac{1}{3}$, 5 items of side $\frac{1}{7}$, 1 item of side $\frac{1}{41}$, 35 items of side $\frac{1}{42}$, and $m$ items of side length equal to $\frac{1}{m}$ that cannot be packed together with an additional item of side $\frac{1}{2}$.}
\label{fig:tight-lemma34}
\end{figure}

The algorithms for the two lists $L_1$ and $L_2$ have asymptotic competitive ratios of $\frac{3377129}{1433976}$ and $1$, respectively, when restricted to unit-fraction items. We again apply the observation of Epstein and Levy stated above.

\begin{theorem}
\label{thm:ufsmb}
The algorithm {\rm SMB} has an asymptotic competitive ratio of at most
 $\frac{4811105}{1433976} \approx 3.356$ when all items are unit-fraction squares.
\end{theorem}
\begin{proof}
By that observation, the ratio is at most $\frac{3377129}{1433976} + 1 = \frac{4811105}{1433976} \approx 3.356$.
\end{proof}

\section{Power-Fraction Items} \label{sec:pf}

In this section, we present an algorithm for the case in which the items are squares with side $\frac{1}{2^k}$ for some $k \in \mathbb{Z}_{\geq 0}$. The algorithm is a First-Fit strategy: as items arrive, they are packed into the first bin in which they fit; if they do not fit in any existing bin, a new bin is opened.

To decide whether a multiset $X$ of items fits into a given bin, we use the following procedure. We partition the bin into 4 squares of side $\frac{1}{2}$, keep 3, and recursively partition the fourth. Given a multiset $X$, while there exist 4 items with the same side $s \leq \frac{1}{2}$, we merge them into one item with double the side. Note that this operation does not change the total area of $X$. On termination, every side $\frac{1}{2^j}$ occurs at most 3 times. 

If no item of side $1$ remains, we place the at most three items of each side
$\frac{1}{2^j}$ in the three cells of that side, and the procedure succeeds. If exactly 1 item of side 1 remains and there is no other item, we place it alone in the bin. In every other case, the procedure fails.

With this in mind, we can prove the following lemma.

\begin{lemma}\label[lemma]{lemma:pfarea}
    The procedure above fails to pack $X$ if, and only if, $\area(X) > 1$.
\end{lemma}

\begin{proof}
    Suppose the algorithm fails to pack the items into the bin. This means that, after merging items into larger ones, there is at least one square of side $1$, since otherwise the algorithm could pack the items. Moreover, if the multiset contained only 1 item of side $1$, then the algorithm would succeed; therefore, there must be some additional item, and thus the merged area exceeds 1. Since the area does not change during the merging, we get that $\area(X) > 1$.

    Conversely, if $\area(X) > 1$, the algorithm fails to pack the items into a bin, since no packing exists.
\end{proof}

By \Cref{lemma:pfarea}, a multiset of power-fraction squares fits in a bin if and only if its total area is at most $1$. The problem is therefore equivalent to one-dimensional dynamic bin packing in which the item sizes are $4^{-j}$, $j \in \mathbb{Z}_{\ge 0}$. In particular, these sizes are unit-fractions, so the bound of $2.4842$ obtained by Han et al.~\cite{HAN2010} for First-Fit on one-dimensional unit-fraction items applies verbatim to our instances.  We show that it can be improved to $\frac{577087}{261120} < 2.211$ by exploiting the fact that the sizes are powers of four. This follows since the use of only items that are powers of four allows for better bounds on the $\fgamma{x}{y}$ function, as we will show later.  We follow the recursion of Chan et al.~\cite{CHAN2008} on the sequence $(b_i, r_i)$ and use these exact values in the case analysis of \Cref{app:power-fraction-cases}.

First, we define a function $\fgamma{x}{y}$, which represents the minimum load that a bin must have if an item of side $\frac{1}{y}$ does not fit and the bin contains items of types $T_x = \{2^j : 0 \leq j \leq \log_2 x\}
= \{1, 2, 4, \ldots, x\}$. Formally,

\begin{equation*}
    \fgamma{x}{y} = \min\left\{ \sum_{i \in T_x} \frac{n_i}{i^2}
    \;\middle|\; n_i \in \mathbb{Z}_{\geq 0} \text{ for all } i \in T_x,
    \ \sum_{i \in T_x} \frac{n_i}{i^2} > 1 - \frac{1}{y^2} \right\}.
\end{equation*}

Some useful facts about this function are that $\fgamma{x}{y} \geq 1 - \frac{1}{y^2}$, $\fgamma{x}{1} = \frac{1}{x^2}$, and $\fgamma{x}{x} = 1$. The first inequality follows directly from the definition of $\gamma$. The second follows from the fact that, since the item to be inserted has side length 1, the smallest possible area of items already inserted is just one item of side $\frac{1}{x}$. Finally, the third fact follows from the observation that every item type has area $4^{-j}$ for some $j$ with $2^j \leq x$, so every achievable load is an integer multiple of $\frac{1}{x^2}$. The smallest such multiple exceeding $1 - \frac{1}{x^2}$ is 1, and a load of exactly 1 is attained by $x^2$ items of side $\frac{1}{x}$; hence $\fgamma{x}{x}=1$.

Number the bins $1, 2, \dots$ in the order in which First-Fit opens them, so that
the last bin opened has index equal to the number of bins used. Now we define the ordered pairs $(b_i, r_i)$ as follows: $b_1$ is the number of bins opened by the First-Fit algorithm, and $\frac{1}{r_1}$ is the side of the smallest item ever placed in bin $b_1$. The next pairs are defined recursively: $b_{i+1}$ is the largest index $b < b_i$ such that some item of side smaller than $\frac{1}{r_i}$ is ever placed in bin $b$, and $\frac{1}{r_{i+1}}$ is the side of the smallest item ever placed in bin $b_{i+1}$. If it is not possible to define another ordered pair, we stop the construction. We denote by $k$ the index of the last ordered pair for a sequence $\sigma$. Note that, by construction, $r_1 < r_2 < \dots < r_k$, and these values are powers of 2.

Now consider the moment when an item of side $\frac{1}{r_k}$ is placed in bin $b_k$. It is easy to see that the total load at this moment is
$$
l_k \geq \frac{1}{r_k^2} + (b_k - 1)\fgamma{r_k}{r_k},
$$

\noindent as bins $1, \dots, b_k - 1$, at that instant, contain no item of side smaller than $\frac{1}{r_k}$, and each rejected the item, so each has load at least $\fgamma{r_k}{r_k}$. Combining this with the load incurred by the item of side $\frac{1}{r_k}$, we obtain the desired bound.

Now consider some time when an item of side $\frac{1}{r_i}$ is being placed in bin $b_i$. The load at this moment is
$$
l_i \geq \frac{1}{r_i^2} + (b_i - b_{i+1} - 1)\fgamma{r_i}{r_i} + \sum_{p = i + 1}^{k - 1} (b_p - b_{p+1}) \fgamma{r_p}{r_i} + b_k \fgamma{r_k}{r_i}.
$$

We start by defining the maximum load $l = \max_i l_i$, which is well defined, as $k$ is finite, since $b_{i + 1} < b_i$, so $k \leq b_1$. We seek inequalities of the form $b_1 \leq c_1 l + c_2$. Note that, since $l \leq \opt$, such an inequality would imply an asymptotic competitive ratio of $c_1$. The observation that $l \leq \opt$ comes from the fact that, at that point, the total area is equal to the load, so an offline algorithm has to use at least $\lceil l \rceil$ bins. The analysis involves multiple cases based on the values of $k$ and the sequence $(r_1, r_2, \ldots, r_k)$. The detailed calculations for all cases are provided in Appendix~\ref{app:power-fraction-cases}. These cases cover all possibilities.  Indeed, if $k = 1$ we are in Case~1, and if $k \ge 2$ and $r_1 \ge 2$ we are in Case~3, so assume $k \ge 2$ and $r_1 = 1$. If $k = 2$ we are in Case~2, so assume $k \ge 3$.  As $r_i$ are strictly increasing powers of two, $r_2 = 2$ or $r_2 \ge 4$; in the latter case we are in Case~5.  If $r_2 = 2$ and $k = 3$ we are in Case~4, so assume $k \ge 4$; again $r_3 = 4$ or $r_3 \ge 8$, and in the latter case we are in Case~7.  Finally, if $r_3 = 4$ then $k = 4$ gives Case~6 and $k \ge 5$ gives Case~8.

\begin{table}[h]
\centering
\renewcommand{\arraystretch}{1.4}
\begin{tabular}{@{}cp{5cm}c@{}}
\toprule
\textbf{Case} & \textbf{Conditions} & \textbf{Competitive Ratio} \\ 
\midrule
1 & $k = 1$ & 1 \\ [0.2em]
2 & $k = 2, r_1 = 1$ & $2$ \\ [0.2em]
3 & $k \geq 2, r_1 \geq 2$ & $\frac{19}{15} \approx 1.267$ \\[0.2em]
4 & $k = 3, r_1 = 1, r_2 = 2$ & $\frac{35}{16} \approx 2.188$ \\[0.2em]
5 & $k \geq 3, r_1 = 1, r_2 \geq 4$ & $\frac{130}{63} \approx 2.063$ \\[0.2em]
6 & $k = 4, r_1 = 1, r_2 = 2, r_3 = 4$ & $\frac{2261}{1024} \approx 2.208$ \\[0.2em]
7 & $k \geq 4, r_1 = 1, r_2 = 2, r_3 \geq 8$ & $\frac{8953}{4080} \approx 2.194$ \\[0.2em]
8 & $k \geq 5, r_1 = 1, r_2 = 2, r_3 = 4$ & $\frac{577087}{261120} \approx 2.211$ \\
\bottomrule
\end{tabular}
\caption{Competitive ratios for different cases in the power-fraction analysis. The worst-case ratio of $\frac{577087}{261120} \approx 2.211$ occurs when $k \geq 5, r_1=1, r_2=2, r_3=4$.}
\label{tab:power-fraction-cases}
\end{table}

As shown in Table~\ref{tab:power-fraction-cases}, the worst case occurs when
$k \geq 5$, $r_1 = 1$, $r_2 = 2$ and $r_3 = 4$. This yields the following result.

\begin{theorem}
\label{thm:pf}
The First-Fit algorithm has an asymptotic competitive ratio of at most
$\frac{577087}{261120} < 2.211$ for dynamic packing of power-fraction squares.
\end{theorem}

The lower bound of $2.2307$ of Epstein and Levy uses squares of side $\frac{1}{3}$ and hence does not apply here. Their lower bound of $2$ for dynamic $d$-dimensional cube packing (in particular squares)~\cite[Theorem~9]{EpsteinL10}, however, uses only squares of sides $\frac{1}{F}$ and $1$, where $F$ is a large integer; taking $F$ to be a power of two gives the following:

\begin{proposition}
    No algorithm has an asymptotic competitive ratio smaller than 2 for the dynamic packing of power-fraction squares.
\end{proposition}

\section{Conclusion}

We presented improved algorithms and analyses for dynamic bin packing of square items.  For arbitrary squares, the algorithm SMB has an asymptotic competitive ratio of at most $\frac{1003}{256} < 3.918$, improving the bound of $4.2154$ of Epstein and Levy~\cite{EpsteinL10}, which had remained the best-known bound for this case.  For unit-fraction and power-fraction squares we obtained $\frac{4811105}{1433976} < 3.356$ and $\frac{577087}{261120} < 2.211$, respectively. For the unit-fraction case, the previous best bound was that of Burcea, Wong and Yung~\cite{BWY2013}; for the power-fraction case, it was the bound of Han et al.~\cite{HAN2010} for one-dimensional unit-fraction items, which applies to power-fraction squares by \Cref{lemma:pfarea}. \Cref{tab:summary} summarizes the known and new results on the competitive ratios for the corresponding problems. In this table, 
\textbf{ub} (resp. \textbf{lb}) stands for \emph{upper bound} (resp. \emph{lower bound}).

\begin{table}[h]
\centering
\renewcommand{\arraystretch}{1.4}
\begin{tabular}{@{}lcccc@{}}
\toprule
\textbf{Item class} &  Old \textbf{ub} & New \textbf{ub} (this paper) &  \textbf{lb} (best known)  & New \textbf{ub}$/$\textbf{lb} \\
\midrule
Arbitrary squares      & $4.2154$~\cite{EpsteinL10} & $\frac{1003}{256} < 3.918$        & $2.2307$~\cite{EpsteinL10} & $ \approx 1.76$ \\
Unit-fraction squares  & $3.9654$~\cite{BWY2013} & $\frac{4811105}{1433976} < 3.356$ & $2.2307$~\cite{EpsteinL10} & $ \approx 1.50$ \\
Power-fraction squares & $2.4842$~\cite{HAN2010}   & $\frac{577087}{261120} < 2.211$   & $2$ \cite{EpsteinL10}    & $ \approx 1.11$ \\
\bottomrule
\end{tabular}
\caption{Previous and current competitive ratios.}
\label{tab:summary}
\end{table}

In our analysis, we showed that our area bounds in \Cref{sec:squares} and \Cref{sec:uf} are tight. Any improvement on the ratios above must therefore come from a different algorithm, or from a finer way of combining these bounds.  We conclude with two open problems.

\paragraph{Lower bounds on the SM and SMB algorithms.} While our analysis provided an upper bound, it is still open how tight this upper bound is. Would it be possible to provide a lower bound for our proposed algorithm in both the unit-fraction and the general-squares settings?

\paragraph{Lower bounds on the general case.} Our analysis showed that our algorithm can obtain a better bound when dealing with unit-fraction items. However, the best currently known lower bound uses only unit-fraction items. Would it be possible to design better lower bounds for the general case?

\bibliography{mybibs}  

\medskip

\hrule

\medskip

\appendix

\section{APPENDIX -- Power-Fraction Case Analysis}\label[appendix]{app:power-fraction-cases}

This appendix provides the detailed calculations for all cases in the power-fraction analysis presented in Section \ref{sec:pf}. Throughout, $c_2$ denotes an absolute constant whose value does not affect the asymptotic competitive ratio.

\subsection{Case 1: $k = 1$.}

Since the construction stops at $k = 1$, no bin with index smaller than $b_1$
contains an item of side smaller than $\frac{1}{r_1}$, and each of them
rejected the arriving item. Hence
\[
l_1 \geq \frac{1}{r_1^2} + (b_1 - 1)\fgamma{r_1}{r_1} \geq b_1 - 1
\implies l + 1 \geq b_1 .
\]

\subsection{Case 2: $k = 2,\ r_1 = 1$.}

\begin{align*}
\begin{split}
l_2 &\geq \frac{1}{r_2^2} + (b_2 - 1)\fgamma{r_2}{r_2} \\
&\geq (b_2 - 1) \implies l + 1 \geq b_2.
\end{split}
\end{align*}

and

\begin{align*}
\begin{split}
l_1 &\geq \frac{1}{r_1^2} + (b_1 - b_2 - 1)\fgamma{r_1}{r_1} + b_2\fgamma{r_2}{r_1} \\
&\geq b_1 - b_2 \implies l + b_2 \geq b_1.
\end{split}
\end{align*}

Combining the two inequalities yields $2l + 1 \geq b_1$.

\subsection{Case 3: $k \geq 2, r_1 \geq 2$.}

If $k = 2$,

\begin{align} \label{A31}
    l_2 & \geq \frac{1}{r_2^2} + (b_2 - 1)\fgamma{r_2}{r_2} 
    \geq (b_2 - 1)\left(1 - \frac{1}{r_2^2}\right)
    \geq (b_2 - 1)\frac{15}{16} \nonumber\\
    & \implies \frac{16}{15}l + 1 \geq b_2.
\end{align}

If $k \geq 3$,
\begin{align} \label{A32}
    l_2 & \geq \frac{1}{r_2^2} 
    + (b_2 - b_3 - 1)\fgamma{r_2}{r_2}
    + \sum_{p = 3}^{k - 1}(b_p - b_{p+1})\fgamma{r_p}{r_2}
    + b_k \fgamma{r_k}{r_2} \nonumber \\
    & \geq (b_2 - 1)\left(1 - \frac{1}{r_2^2}\right)
    \geq (b_2 - 1)\frac{15}{16}
    \implies \frac{16}{15}l + 1 \geq b_2.
\end{align}

\begin{align} \label{A33}
    l_1 & \geq \frac{1}{r_1^2} 
    + (b_1 - b_2 - 1)\fgamma{r_1}{r_1}
    + \sum_{p = 2}^{k - 1}(b_p - b_{p+1})\fgamma{r_p}{r_1}
    + b_k\fgamma{r_k}{r_1} \nonumber \\
    & \geq b_1 - b_2 - \frac{3}{4} + b_2 \cdot \frac{3}{4}
    \implies l + \frac{1}{4} b_2 + \frac{3}{4} \geq b_1.
\end{align}

Combining Equations \eqref{A32} or \eqref{A31} and \eqref{A33}, we obtain  
$l + \frac{1}{4} b_2 + \frac{3}{4} \geq b_1$,  
and hence $\frac{19}{15}l + c_2 \geq b_1$.

\subsection{Case 4: $k = 3, r_1 = 1, r_2 = 2$.}

\begin{align} \label{A41}
    l_3 & \geq \frac{1}{r_3^2} + (b_3 - 1)\fgamma{r_3}{r_3} 
    \geq (b_3 - 1) 
    \implies l + 1 \geq b_3.
\end{align}

\begin{align} \label{A42}
    l_2 & \geq \frac{1}{r_2^2} 
    + (b_2 - b_3 - 1)\fgamma{r_2}{r_2}
    + b_3 \fgamma{r_3}{r_2} \nonumber \\
    & \geq \frac{1}{4} + (b_2 - b_3 - 1) + b_3\frac{3}{4} \nonumber \\
    & = b_2 - \frac{1}{4}b_3 - \frac{3}{4}
    \implies l + \frac{1}{4}b_3 + \frac{3}{4} \geq b_2
    \implies \frac{5}{4}l + 1 \geq b_2.
\end{align}

\begin{align} \label{A43}
    l_1 & \geq 1 
    + (b_1 - b_2 - 1)\fgamma{1}{1}
    + (b_2 - b_3)\fgamma{2}{1}
    + b_3\fgamma{r_3}{1} \nonumber \\
    & \geq b_1 - \left(1 - \frac{1}{4}\right)b_2 - \frac{1}{4}b_3
    \implies l + \frac{3}{4}b_2 + \frac{1}{4}b_3 \geq b_1.
\end{align}

Combining Equations \eqref{A41}, \eqref{A42}, and \eqref{A43}, we obtain  
$\frac{35}{16}l + c_2 \geq b_1$.

\subsection{Case 5: $k \geq 3, r_1 = 1, r_2 \geq 4$.}

Again, we obtain

\begin{align} \label{A51}
    l_3 & \geq (b_3 - 1)\left(1 - \frac{1}{r_3^2}\right)
    \geq (b_3 - 1)\frac{63}{64}
    \implies \frac{64}{63}l + 1 \geq b_3.
\end{align}

\begin{align} \label{A52}
    l_2 & \geq \frac{1}{r_2^2} 
    + (b_2 - b_3 - 1)\fgamma{r_2}{r_2}
    + \sum_{p = 3}^{k-1}(b_p - b_{p+1})\fgamma{r_p}{r_2}
    + b_k\fgamma{r_k}{r_2} \nonumber\\
    & \geq b_2 - \frac{1}{16}b_3 - 1 
    \implies l + \frac{1}{16}b_3 + 1 \geq b_2 
    \implies \frac{67}{63}l + \frac{17}{16} \geq b_2.
\end{align}

\begin{align} \label{A53}
    l_1 & \geq 1 + (b_1 - b_2 - 1)\fgamma{1}{1}
    + \sum_{p = 2}^{k-1}(b_p - b_{p+1})\fgamma{r_p}{1}
    + b_k\fgamma{r_k}{1}\nonumber\\
    & \geq b_1 - \left(1 - \frac{1}{r_2^2}\right)b_2 - \frac{1}{r_2^2}b_3
    \geq b_1 - b_2 .
\end{align}

Combining Equations \eqref{A51}, \eqref{A52}, and \eqref{A53}, we obtain  
\[
\frac{130}{63}l + c_2 \geq b_1.
\]

\subsection{Case 6: $k = 4, r_1=1, r_2 = 2, r_3 =4$.}

\begin{align} \label{A61}
    l_4 \geq (b_4 - 1) \implies l + 1 \geq b_4.
\end{align}

\begin{align} \label{A62}
    l_3 & \geq \frac{1}{16} + (b_3 - b_4 - 1)\fgamma{r_3}{r_3} + b_4 \fgamma{r_4}{r_3} \geq b_3 - \frac{1}{16}b_4 - \frac{15}{16} \nonumber \\ & \implies \frac{17}{16}l + 1 \geq b_3.
\end{align}

\begin{align} \label{A63}
    l_2 & \geq \frac{1}{4} + (b_2 - b_3 - 1) \fgamma{r_2}{r_2} + (b_3 - b_4) \fgamma{r_3}{r_2} + b_4 \fgamma{r_4}{r_2} \nonumber\\
     & \geq b_2 - (1 - \frac{13}{16})b_3 - (\frac{13}{16} - \frac{3}{4})b_4 -\frac{3}{4} \implies l + \frac{3}{16}b_3 + \frac{1}{16}b_4 + \frac{3}{4} \geq b_2 \nonumber \\
     & \implies \frac{323}{256}l + 1 \geq b_2.
\end{align}

\begin{align} \label{A64}
    l_1 & \geq b_1 - b_2 + (b_2 - b_3) \fgamma{2}{1} + (b_3 - b_4) \fgamma{4}{1} + b_4 \fgamma{r_4}{1} \nonumber \\
    & \geq b_1 - (1 - \frac{1}{4})b_2 - (\frac{1}{4} - \frac{1}{16}) b_3 - (\frac{1}{16} - \frac{1}{r_4^2})b_4 \nonumber \\
    & \geq b_1 - (1 - \frac{1}{4})b_2 - (\frac{1}{4} - \frac{1}{16}) b_3 - \frac{1}{16}b_4.
\end{align}

Combining Equations \eqref{A61}, \eqref{A62}, \eqref{A63}, and \eqref{A64}, we obtain that $\frac{2261}{1024}l + c_2 \geq b_1$.

\subsection{Case 7: $k \geq 4, r_1 =1, r_2 = 2, r_3 \geq 8$.}

Similarly to what was done previously, we obtain

\begin{align} \label{A71}
    l_4 \geq (b_4 - 1) \frac{255}{256} \implies l\frac{256}{255} + 1 \geq b_4.
\end{align}

\begin{align} \label{A72}
    l_3 \geq (b_3 - b_4 - 1) + b_4(1 - \frac{1}{r_3^2}) \geq b_3 - \frac{1}{64}b_4 - 1 \implies \frac{259}{255}l + \frac{65}{64} \geq b_3 .
\end{align}

\begin{align} \label{A73}
    l_2 & \geq \frac{1}{r_2^2} + (b_2 - b_3 - 1) \fgamma{r_2}{r_2} + \sum_{p = 3}^{k-1} (b_p - b_{p + 1}) \fgamma{r_p}{r_2} + b_k \fgamma{r_k}{r_2} \nonumber \\
    & \geq (b_2 - b_3 - 1) + \frac{3}{4}b_3 = b_2 - \frac{1}{4}b_3 - 1 \implies \frac{1279}{1020}l + \frac{321}{256} \geq b_2.
\end{align}

\begin{align} \label{A74}
    l_1 & \geq b_1 - b_2 + (b_2 - b_3) \fgamma{r_2}{1} +   (b_3 - b_4) \fgamma{r_3}{1} \nonumber \\
    & \geq b_1 - (1 - \frac{1}{4})b_2 -\frac{1}{4}b_3. 
\end{align}

Combining Equations \eqref{A71}, \eqref{A72}, \eqref{A73}, and \eqref{A74}, we obtain that $\frac{8953}{4080}l + c_2 \geq b_1$.

\subsection{Case 8: $k \geq 5, r_1=1, r_2=2, r_3=4$.}

Again, we have that

\begin{align} \label{A81}
l_5 & \geq (b_5 - 1) \frac{255}{256} \implies \frac{256}{255}l + 1 \geq b_5.
\end{align}

\begin{align} \label{A82}
    l_4 & \geq (b_4 - b_5 - 1) \fgamma{r_4}{r_4} + \sum_{p = 5}^{k - 1} (b_p - b_{p + 1}) \fgamma{r_p}{r_4} + b_k \fgamma{r_k}{r_4} \nonumber \\
    & \geq b_4 - \frac{1}{r_4^2} b_5 - 1 \geq b_4 - \frac{1}{64}b_5 -1 \implies \frac{259}{255}l + \frac{65}{64} \geq b_4.
\end{align}

\begin{align} \label{A83}
    l_3 & \geq (b_3 - b_4 - 1) \fgamma{r_3}{r_3} + \sum_{p = 4}^{k - 1} (b_p - b_{p + 1}) \fgamma{r_p}{r_3} + b_k \fgamma{r_k}{r_3} \nonumber \\
    & \geq b_3 - \frac{1}{16}b_4 - 1 \implies \frac{4339}{4080}l + \frac{1089}{1024} \geq b_3.
\end{align}
    
\begin{align} \label{A84}
     l_2 & \geq (b_2 - b_3 - 1) \fgamma{r_2}{r_2} + \sum_{p = 3}^{k - 1} (b_p - b_{p + 1}) \fgamma{r_p}{r_2} + b_k \fgamma{r_k}{r_2} \nonumber \\
    & \geq b_2 - (1- \frac{13}{16})b_3 - (\frac{13}{16} - \frac{3}{4})b_4 -1 \implies \frac{82441}{65280} l + \frac{20691}{16384} \geq b_2.
\end{align}

\begin{align} \label{A85}
    l_1 & \geq b_1 - b_2 + (b_2 - b_3) \frac{1}{4} + (b_3 - b_4) \frac{1}{16} + (b_4 - b_5) \frac{1}{r_4^2} \nonumber \\
    &\geq b_1 - (1 - \frac{1}{4})b_2 - (\frac{1}{4} - \frac{1}{16})b_3 - \frac{1}{16}b_4.
\end{align}

Combining Equations \eqref{A81}, \eqref{A82}, \eqref{A83}, \eqref{A84}, and \eqref{A85}, we obtain that $\frac{577087}{261120}l + c_2 \geq b_1$.

\end{document}